\documentclass[conference]{IEEEtran}

\usepackage{bm,balance}
\usepackage{verbatim}
\usepackage{algorithm}
\usepackage{soul,xcolor}
\usepackage[T1]{fontenc}
\usepackage{algcompatible}
\usepackage[normalem]{ulem}
\usepackage{multirow,enumitem}
\usepackage{graphicx}

\usepackage{url}
\usepackage{bbm}
\usepackage{cite}
\usepackage{array}
\usepackage{ifthen}
\usepackage{xspace}
\usepackage{dsfont}
\usepackage{amsmath}
\usepackage{amssymb}
\usepackage{multicol}
\usepackage{amsfonts}
\usepackage{mathrsfs}
\usepackage{booktabs}
\usepackage{caption}
\usepackage{subcaption}
\usepackage{algpseudocode}

\usepackage{esint}
\usepackage{amsthm}
\usepackage{cancel}
\usepackage{siunitx}
\usepackage{setspace}
\usepackage{hyperref}
\usepackage{makecell}
\usepackage{footnote}

\theoremstyle{plain}
\newtheorem{lemma}{Lemma}
\newtheorem{theorem}{Theorem}
\newtheorem{assumption}{Assumption}

\theoremstyle{definition}

\newtheorem{corollary}{Corollary}

\theoremstyle{remark}

\newcommand{\proposed}{NCOTA-DGD}
\newcommand{\secref}[1]{Sec.~\ref{#1}}
\newcommand{\e}{\boldsymbol{\epsilon}}
\newcommand{\W}{\hat{\boldsymbol{\Omega}}}

\DeclareMathOperator*{\argmin}{arg\,min}

\title{Decentralized Federated Learning via\\ Non-Coherent Over-the-Air Consensus}

\author{Nicol\`o Michelusi
\thanks{N. Michelusi is with the School of Electrical, Computer and Energy Engineering, Arizona State University.}
\thanks{An extended version of this work appears in \cite{techreport}.}
\thanks{This research has been funded in part by NSF under grant CNS-2129615.}
}

\IEEEoverridecommandlockouts
\begin{document}

\maketitle
\setulcolor{red}
\setul{red}{2pt}
\setstcolor{red}
\begin{abstract}
This paper presents \proposed, a Decentralized Gradient Descent (DGD) algorithm that 
combines local gradient descent with a novel Non-Coherent Over-The-Air (NCOTA) consensus scheme to solve 
distributed machine-learning problems over wirelessly-connected systems.
\proposed\ leverages the
waveform superposition properties of the wireless channels:
it enables simultaneous transmissions under half-duplex constraints, by mapping local optimization signals to a mixture of preamble sequences, and consensus
via  non-coherent combining at the receivers.
\proposed\ operates without channel state information at transmitters and receivers, and leverages the average channel pathloss to mix signals,
 without explicit knowledge of the mixing weights (typically known in
 consensus-based optimization algorithms).
It is shown both theoretically and numerically that, for smooth and strongly-convex problems with fixed consensus and learning stepsizes, the updates of \proposed\ converge in Euclidean distance to the global optimum with rate
$\mathcal O(K^{-1/4})$ for a target of $K$ iterations.
\proposed\ is evaluated numerically  over a logistic regression problem, showing
faster convergence vis-\`a-vis running time than  implementations of the classical DGD algorithm over  digital and analog orthogonal channels.
\end{abstract}

\vspace{-2mm}
\section{Introduction}\noindent
\emph{Federated learning} (FL) \cite{McMahan} 
has emerged as a new paradigm to alleviate the communication burden and  privacy concerns associated with the  transmission of raw data to a  ML server, by leveraging decentralized computational and communication resources  at the edge of the network.
Typically, it aims to solve
\begin{align}
\label{global}
{\mathbf w}^*=\argmin_{\mathbf w\in\mathbb R^d}
\ F(\mathbf{w})\triangleq\frac{1}{N}\sum_{i=1}^Nf_i(\mathbf w) \tag{P}  
\end{align}
among $N$ edge devices,
where $f_i(\mathbf w)$ is the empirical loss based on the local dataset of node $i$ (known to $i$ alone), and
 $\mathbf w$ is a $d$-dimensional parameter vector, so that $F(\mathbf w)$ represents the  empirical loss over the network. 
Conventional FL solves \eqref{global} based on a client-server architecture, where the $N$ nodes interact with a parameter server (PS, such as a base station) over multiple rounds
 \cite{Lian17}: in each round, edge devices compute local gradients based on a global model broadcast by the PS,
 and transmit them to the PS; the latter aggregates the local gradients and updates the global model via gradient descent.
 
 Yet, in many important scenarios, a PS  may be lacking \cite{8950073},
 or direct  communication with the PS may be challenging due to severe channel propagation conditions.
In these cases, a decentralized learning architecture may be more attractive, in which
the edge devices communicate with each other without the aid of a PS \cite{8950073}.
A renowned algorithm to solve \eqref{global} in this setting is \emph{Decentralized Gradient Descent} (DGD) \cite{Yuan2016}:
at iteration $k$, each node ($i$) 
updates its local optimization signal $\mathbf w_{i,k}$ by
combining a consensus signal $\mathbf c_{i,k}$ (weighted sum of the neighbors' signals) with a local gradient step, as
\begin{align}
\mathbf w_{i,k+1}
&
=\underbrace{\mathbf w_{i,k}+\sum_{j=1}^N\omega_{i,j}(\mathbf w_{j,k}-\mathbf w_{i,k})}_{\triangleq\mathbf c_{i,k}}
-\eta\nabla f_i(\mathbf w_{i,k});
\label{ddd}
\end{align}
here, $\omega_{i,j}$ are non-negative mixing weights
 with $\omega_{i,j}{=}\omega_{j,i}$ and $\sum_{j=1}^N\omega_{i,j}{=}1$
 ($\omega_{i,j}{=}0$ if $i$ and $j$ do not  communicate directly).

Nevertheless, \eqref{ddd}
relies on communications over orthogonal, noise- and error-free links,
and on knowledge of the mixing weights to aggregate the incoming signals.
In many practical scenarios (e.g.,  swarms of UAVs), communications occur over wireless links: 
interference from simultaneous transmissions, fading and noise may preclude the ability to receive error-free signals. Mitigating these sources of errors typically requires: 1) 
centralized coordination of network scheduling and interference management operations; 
 2) channel state information (CSI) to compensate signal fluctuations and link outages due to fading.
Yet, such coordination may be non-trivial  in  wireless decentralized systems, and CSI acquisition may be impaired by pilot contamination and a source of severe overhead. 
 
 This calls for the design of decentralized optimization schemes that operate over (and leverage properties of) wireless channels.
In this paper, we present \proposed, an implementation of \eqref{ddd} over wireless channels subject to noise, fading, and interference.
Its main feature is a \emph{Non-Coherent} (NC)-\emph{Over-the-Air} (OTA) consensus step that
leverages the waveform superposition properties of the wireless channels, without the need for CSI at transmitters and receivers, and without explicit knowledge of the mixing weights.
 We show theoretically and numerically that,  for smooth and strongly-convex problems with suitable constant \emph{consensus} and \emph{learning} stepsizes, the error (Euclidean distance) between the local models and the solution of \eqref{global}
 converges to zero with rate $\mathcal O(K^{-1/4})$ for a target of $K$ iterations. 

Recent OTA-based schemes under client-server \cite{8952884,9014530,9272666,9382114,9042352} and decentralized \cite{9563232,9517780,9322286} FL  rely on accurate CSI and careful power control  to compensate signal fluctuations due to fading.
Focusing on decentralized FL, the works \cite{9563232,9517780,9322286} use graph coloring to break down the network into smaller non-interfering subgraphs, in which one device operates as the PS. This expedient enables the use of techniques developed for client-server FL (including channel inversion) coupled with a suitable consensus enforcing step. 
Yet, these schemes rely on CSI and power control to invert channels, scheduling operations (e.g., graph coloring),
and knowledge of the network structure and mixing weights $\omega_{i,j}$ for consensus. 
 In contrast, \proposed\ overcomes this need
 by using a set of orthogonal preamble sequences to encode signals, coupled with non-coherent combining at the receivers;
 it leverages  the channel pathloss to mix signals, without explicit knowledge of the mixing weights.
The paper \cite{4511469} studies \emph{consensus} over a shared multipath channel.
While it assumes the channels to be \emph{noiseless} and \emph{static}, our work focuses on \emph{decentralized optimization} over \emph{noisy, time-varying} fading channels.
 
The preamble-based technique developed in this paper is inspired by the preamble-based random access scheme 
 developed in  \cite{9815298}  to encode local gradients in \emph{client-server FL}. The scheme therein relies on noise-free downlink and inversion of the average pathloss at the transmitters. 
  Differently from  \cite{9815298}: 1) \proposed\ operates in decentralized settings, with all links subject to fading and noise;
  2) instead of random access-based preamble selection of \cite{9815298}, \proposed\ 
   maps local signals deterministically to a suitable linear combination of preamble sequences;
 3) rather than inverting channels, \proposed\  
  leverages the average pathloss to mix signals.

 This paper is organized as follows. In \secref{sysmo}, we describe \proposed,
 followed by its convergence analysis in \secref{convanalysis}.
In \secref{numres}, we present numerical results, followed by concluding remarks in \secref{conclu}.
An extended version of this work appears in \cite{techreport}, including
complete proofs, a more general and non-trivial analysis with decreasing stepsizes, more general fading models and frequency selective channels,
stochastic gradient descent updates and extensive numerical~evaluations.

\emph{Notation}: All vectors are in column form.
 For vector $\mathbf a$ (boldface),
 $[\mathbf a]_i$ is its $i$th component, and $\Vert\mathbf a\Vert{=}\sqrt{\mathbf a^{\mathrm H}\mathbf a}$ its Euclidean norm.
 For random vector $\mathbf a$, we define $\Vert\mathbf a\Vert_{\mathbb E}{\triangleq}\sqrt{\mathbb E[\Vert\mathbf a\Vert^2]}$
and $\Vert\mathbf a\Vert_{\mathbb E|A}$ when the expectation is conditional on event $A$.
$\mathbf e_m$ is the standard basis vector with $m$th component equal to 1 and 0 otherwise;
 $\mathbf 1$ and $\mathbf 0$ are vectors of 1's and 0's; their dimension is deduced from the context.
$\mathbf I_n$ is the $n\times n$ identity matrix.
$\mathbbm{1}[A]$ is the indicator of  event $A$.
$\mathbf A\otimes \mathbf B$ is the Kronecker product of matrices $\mathbf A$ and $\mathbf B$.

\vspace{-1mm}
\section{System Model and \proposed}
\label{sysmo}
We consider $N$ wirelessly-connected edge devices, solving \eqref{global}
 via a noisy version of DGD in \eqref{ddd}.
We assume that the optimizer of \eqref{global}, $\mathbf w^*$, lies in a
$d$-dimensional sphere $\mathcal W$ of radius $R$, within which the optimization is restricted (without loss of generality, centered at $\mathbf 0$).
For instance, 
 for strongly-convex $F(\cdot)$ with strong-convexity parameter $\mu$,
 as assumed in the convergence analysis of \secref{convanalysis},
  since $\nabla F(\mathbf w^*){=}\mathbf 0$,
 it holds $\Vert\nabla F(\mathbf 0)\Vert{=}\Vert\nabla F(\mathbf 0){-}\nabla F(\mathbf w^*)\Vert{\geq}\mu \Vert\mathbf 0{-}\mathbf w^*\Vert$,
hence $\mathcal W{\equiv}\{\mathbf w{\in}\mathbb R^d{:}\Vert\mathbf w\Vert{\leq}R\}$, where $R{\triangleq}\frac{1}{\mu}\Vert\nabla F(\mathbf 0)\Vert$ may be initially estimated via a consensus phase.
 
 To solve \eqref{global} iteratively,
we divide time into frames of fixed duration $T$. In frame $k$,
node $i$ generates the transmission signal $\mathbf x_{i,k}$ using the \emph{Preamble-based Encoding} procedure of \secref{encoding},
and  transmits it over the wireless channel. Upon
 receiving the signal from the other nodes in the network (\secref{transmission}), it then computes a consensus
 signal using the \emph{Non-Coherent Over-the-Air} procedure of \secref{otacomputation};
finally, it updates the local optimization variable
 by combining the
 consensus signal with local gradient descent (\secref{optimization}).
Due to randomness of
noise and fading,
this procedure induces a stochastic process
defined on a proper probability space; we denote by $\mathcal F_k$ the $\sigma$-algebra consisting of
all signals generated up to frame $k$ excluded, along with $\mathbf w_{i,k}$.
\vspace{-4mm}
\subsection{Preamble-based Encoding}
\vspace{-1mm}
\label{encoding}
Let $\mathcal Z{\equiv}\{\mathbf z_m{\in}\mathbb R^d{:}m{=}1,\dots,M\}$ be a codebook of $M{=}d{+}1$ codewords and $\mathbf Z=[\mathbf z_1,\dots,\mathbf z_M]$ be the $d{\times}M$ matrix with $m$th column equal to $\mathbf z_m$.
These are defined as $\mathbf z_{d+1}{=}-R\mathbf 1$ and
$\mathbf z_m{=}2Rd\mathbf e_m{-}R\mathbf 1,m=1,\dots,d$. With this choice, 
any $\mathbf w\in\mathcal W$ may be represented as a convex combination of $\mathcal Z$.
To see this, define the convex combination vector $\mathbf p\in\mathbb R^M$ as
\begin{align}
\label{mb}
[\mathbf p]_m{=}\frac{[\mathbf w]_m{+}R}{2Rd},\forall m=1,\dots, d,
\  [\mathbf p]_{d+1}{=}1{-}\sum_{m=1}^d [\mathbf p]_m.
\end{align}
Since $\mathbf w\in\mathcal W$ (i.e., $\Vert\mathbf w\Vert{\leq}R$ and $[\mathbf w]_m{\in}[-R,R]$),
it is straightforward to see that $\mathbf p$ takes value from 
the $M$-dimensional probability simplex, i.e. $\mathbf 1^\top\mathbf p=1$ and $[\mathbf p]_m\geq 0,\forall m$, and
\begin{align}
\label{convcomb}
\mathbf w=\sum_{m=1}^M [\mathbf p]_m\mathbf z_m=\mathbf Z\cdot \mathbf p.
\end{align}
Hence, $\mathbf p$ in \eqref{mb} defines the desired convex combination.

 Let $\mathcal U{=}\{\mathbf u_m{\in}\mathbb C^M:m{=}1,\dots,M\}$ be a set of $M$ orthogonal preamble sequences, 
defined as $\mathbf u_m{=}\sqrt{M}\mathbf e_m$.
Then, given its local optimization signal  $\mathbf w_{i,k}$, node $i$ generates
 the convex combination vector $\mathbf p_{i,k}$ via  \eqref{mb}, and
 the transmission signal 
\begin{align}\label{txsig}
\mathbf x_{i,k}=\sqrt{E}\sum_{m=1}^{M}\sqrt{[\mathbf p_{i,k}]_m}\mathbf u_m,
\end{align}
with average energy per sample $E=\Vert\mathbf x_{i,k}\Vert^2/M$. 

\vspace{-3mm}
\subsection{Transmission over the wireless channel}
\label{transmission}
Each node  then transmits its signal $\mathbf x_{i,k}$ over the wireless channel.
We assume Rayleigh flat fading channels
 $h_{i,j}^{k}\sim\mathcal {CN}(0,\Lambda_{i,j})$ between transmitting node $j$ and receiving node $i$ in frame $k$,
where $\Lambda_{i,j}$ is the  large-scale pathloss. We assume $\Lambda_{i,j}=\Lambda_{j,i}$ (channel reciprocity).
Furthermore, $h_{i,j}^{k}$ is i.i.d. over $k$, and independent across $i$-$j$.

We assume that the nodes operate under a half-duplex constraint.
We thus divide each frame of duration $T$ into $2$ slots; 
each node is assigned to transmit in only one of the $2$ slots, and operates in receive mode in the other slot. 
This assignment is kept fixed during the entire optimization session, and may be
done randomly by each node.
Let $\mathcal N_i$ be the set of nodes that transmit when node $i$ is in receive mode.
Node $i$ thus receives the signal 
\begin{align}
\label{rxsignal}
\mathbf y_{i,k}=\sum_{j\in\mathcal N_i}h_{i,j}^{k}\mathbf x_{j,k}+\mathbf n_{i,k},
\end{align}
where $\mathbf n_{i,k}{\sim}\mathcal {CN}(\mathbf 0,\sigma^2\mathbf I_M)
$ is AWGN noise with variance $\sigma^2$. $\mathbf y_{i,k}$
 is then correlated with the $M$ preamble sequences as
\begin{align}
&r_{i,m}^{k}
=
\frac{\mathbf u_m^{\mathrm H}\mathbf y_{i,k}}{\sqrt{E}\Vert\mathbf u_m\Vert^2}
=
\sum_{j\in\mathcal N_i}h_{i,j}^{k}\sqrt{[\mathbf p_{j,k}]_m}+n_{i,m}^{k};
\label{rxmatched}
\end{align}
$n_{i,m}^{k}\triangleq\mathbf u_m^{\mathrm H}\mathbf n_{i,k}/(\sqrt{E}\Vert\mathbf u_m\Vert^2)\sim\mathcal{CN}(0,\sigma^2/(M\cdot E))$ is the equivalent noise,
i.i.d. over $i,k,m$, due to the orthogonality of preamble sequences.
Since $h_{i,j}^{k}\sim\mathcal {CN}(0,\Lambda_{i,j})$,
one can see that $|r_{i,m}^{k}|^2|\mathcal F_k$ has exponential distribution with mean
\begin{align}
\label{Erm}
&\mathbb E[|r_{i,m}^{k}|^2|\mathcal F_k]
=
\sum_{j\in\mathcal N_i}\Lambda_{i,j}[\mathbf p_{j,k}]_m
+\frac{\sigma^2}{M\cdot E},
\end{align}
\vspace{-3mm}\\
a fact exploited to build the consensus signal in the next step.
\vspace{-6mm}
\subsection{Non-Coherent Over-the-Air Consensus}
\label{otacomputation}
At the end of frame $k$, node $i$ computes the \emph{consensus} signal
\begin{align}
\label{dik}
&\mathbf d_{i,k}=
\sum_{m=1}^M\Big(|r_{i,m}^{k}|^2-\frac{\sigma^2}{M\cdot E}\Big)(\mathbf z_m-{\mathbf w}_{i,k}).
\end{align}
Using \eqref{Erm}, it is straightforward to see that
\begin{align}
\label{Ed}
\mathbb E[{\mathbf d}_{i,k}|\mathcal F_k]
=
\sum_{j\in\mathcal N_i}\Lambda_{i,j}\sum_{m=1}^M[\mathbf p_{j,k}]_m(\mathbf z_m-{\mathbf w}_{i,k}).
\end{align}
\vspace{-2mm}\\
Furthermore, let
$\Lambda^*{\triangleq}\max_i\sum_{j\in\mathcal N_i}\Lambda_{i,j},$
and define the mixing weights
$\omega_{i,j}{=}\frac{\Lambda_{i,j}}{\Lambda^*}{\cdot}\mathbbm{1}[j{\in}\mathcal N_i]$ for $i\neq j$
and $\omega_{i,i}=1-\sum_{j\neq i}\omega_{i,j}$. Using \eqref{convcomb}, we can then rewrite \eqref{Ed} as
\begin{align}
\mathbb E[{\mathbf d}_{i,k}|\mathcal F_k]
=
\Lambda^*\sum_{j=1}^N\omega_{i,j}(\mathbf w_{j,k}-{\mathbf w}_{i,k}).
\end{align}
\vspace{-2mm}\\
We can see that $\mathbf d_{i,k}$ represents a noisy consensus signal reminiscent of \eqref{ddd}, which aggregates and averages the local signals of all the other nodes in the network, with mixing weights $\omega_{i,j}$ proportional to the average pathloss $\Lambda_{i,j}$.
\vspace{-2mm}
\subsection{Local optimization state update}
\label{optimization}
Finally, node $i$ updates $\mathbf w_{i,k}$
 by combining the consensus signal with a local gradient descent (computed in parallel with transmission and reception), followed by a projection onto $\mathcal W$, yielding the \proposed\ update
\begin{align}
\label{updateeq}
\mathbf w_{i,k+1}=\Pi[\mathbf w_{i,k}+\gamma{\mathbf d}_{i,k}-\eta\nabla f_i(\mathbf w_{i,k})],
\end{align}
where $\gamma>0$ and $\eta>0$ are \emph{consensus} and \emph{learning} stepsizes, respectively.
As shown in \secref{convanalysis}, these need to be chosen suitably, in order to mitigate the impact of fading and noise.
The projection operator $\Pi[\mathbf a]$ onto the sphere $\mathcal W$  is defined as
$$\Pi[\mathbf a]=\arg\min_{\mathbf w\in\mathcal W}\Vert\mathbf w-\mathbf a\Vert=
\begin{cases}
\mathbf a & \text{if }\Vert\mathbf a\Vert\leq R,\\
\frac{R}{\Vert\mathbf a\Vert}\mathbf a & \text{otherwise},\\
\end{cases}
$$
and guarantees that
 $\mathbf w_{i,k}\in\mathcal W,\forall i,\forall k$.
The process  described in \secref{encoding} to \secref{optimization} is then repeated in frame $k+1$ with the new local optimization variable $\mathbf w_{i,k+1}$, and so on.
A sketch of the overall \proposed\ algorithm is shown below:

\begin{algorithm} [b]
\caption{\proposed}\label{A1}
    \begin{algorithmic}[1]
        \scriptsize
        \State \textbf{Initialization}: $\mathbf w_{i,0}\in\mathcal W$; transmit slot assignment for each node $i=1,\dots, N$;
        \For{$k=0,1,\dots $, each node}
        \Procedure{}{at node $i$, given $\mathbf w_{i,k}$}
                 \State Compute convex combination $\mathbf p_{i,k}$ via \eqref{mb};
         \State Generate transmission signal $\mathbf x_{i,k}$ via \eqref{txsig};
         \State Transmit $\mathbf x_{i,k}$ on the assigned transmit slot;
         \State Receive $\mathbf y_{i,k}$ on the assigned receive slot (see \eqref{rxsignal});
         \State Compute
         $r_{i,m}^{k}
=
\mathbf u_m^{\mathrm H}\mathbf y_{i,k}/(\sqrt{E}\Vert\mathbf u_m\Vert^2),\ \forall m=1,\dots, M$
\State \ \ \ 
and \emph{consensus} signal
$\mathbf d_{i,k}
$, as in \eqref{dik}; 
\State Update
$
\mathbf w_{i,k+1}=\Pi[\mathbf w_{i,k}+\gamma{\mathbf d}_{i,k}-\eta\nabla f_i(\mathbf w_{i,k})]
$ as in \eqref{updateeq}.
    \EndProcedure
    \State $k\gets k+1$
         \EndFor
    \end{algorithmic}
\end{algorithm}

We now express  \proposed\  as a noisy version of DGD in \eqref{ddd}.
Let
$
\e_{i,k}={\mathbf d}_{i,k}
-\mathbb E[{\mathbf d}_{i,k}|\mathcal F_k]
$
be the error due to fading and noise. Using \eqref{Ed},
we can then rewrite \eqref{updateeq} as
\vspace{-2mm}
\begin{align}
\label{eqx_NF}
\nonumber
\mathbf w_{i,k+1}=\Pi\Big[&
{\mathbf w}_{i,k}
+\gamma\Lambda^*\sum_{j=1}^N\omega_{i,j}(\mathbf w_{j,k}-{\mathbf w}_{i,k})\\&
-\eta\nabla f_i(\mathbf w_{i,k})+\gamma\e_{i,k}\Big]. 
\end{align}
\vspace{-5mm}\\
A few observations are in order:
\begin{enumerate}[leftmargin=0.2in]
\item The mixing weights satisfy $\omega_{i,j}\geq 0,\forall i, j$ and
$\omega_{i,j}=\omega_{j,i}$, since $\Lambda_{i,j}=\Lambda_{j,i}$ (channel reciprocity)
and $\{j\in\mathcal N_i\}\Leftrightarrow\{ i\in\mathcal N_j\}$.
Hence,  $[\boldsymbol{\Omega}]_{i,j}\triangleq\omega_{i,j}$ is a symmetric, doubly-stochastic mixing matrix (as commonly assumed in consensus-based optimization \cite{Yuan2016})
induced by the large-scale propagation conditions of the channel.
\item When $\gamma{=}(\Lambda^*)^{-1}$, $\e_{i,k}{=}\mathbf 0$,
and neglecting the projection operation, \eqref{eqx_NF}
 reduces to the DGD updates \eqref{ddd}.
 Hence, \proposed\ can be interpreted as a 
projected DGD with noisy consensus. The consensus stepsize $\gamma$ helps mitigate the detrimental effect of error propagation due to noise and fading.
Remarkably, unlike \eqref{ddd}, no explicit knowledge of the mixing weights is required in \proposed.
\end{enumerate}
Since each frame includes $2$ slots,  and in each slot $M{=}d{+}1$ samples are transmitted,
the frame duration of \proposed\ is $T=\frac{2(d+1)}{W_{tot}}$, where $W_{tot}$ is 
the bandwidth of the system.
\vspace{-1mm}
\section{Convergence Analysis}
\vspace{-1mm}
\label{convanalysis}
Let $\mathbf W_k=\sum_{i=1}^N\mathbf e_i\otimes \mathbf w_{i,k}$ be the $\mathbf w_{i,k}$-signals stacked  over the network;
similarly, let $\e_k=\sum_{i=1}^N\mathbf e_i\otimes \e_{i,k}$ be the error signals due to fading and noise, stacked over the network.
Let\vspace{-1mm}
\begin{align}\nonumber
f(\mathbf W)=\sum_{i=1}^Nf_i(\mathbf w_i)
\text{ and }\W=\boldsymbol{\Omega}\otimes\mathbf I_d.
\end{align}
\vspace{-3mm}\\
We  then stack the updates \eqref{eqx_NF} as
\begin{align}
\nonumber
\mathbf W_{k+1}{=}\Pi^N\Big[&{\mathbf W}_{k}
{+}\gamma\Lambda^*(\W-\mathbf I_{Nd}){\mathbf W}_{k}-\eta\nabla f(\mathbf W_{k})
+\gamma\e_k\Big],
\end{align}
\vspace{-3mm}\\
where $\Pi^N[\mathbf A]=\arg\min_{\mathbf W\in\mathcal W^N}\Vert\mathbf W-\mathbf A\Vert$ is the projection operator, stacked over the network.
Similarly to \cite{Yuan2016} for the analysis of DGD,
we interpret this update as a noisy centralized projected gradient descent step with stepsize $\eta$ (see \eqref{Lyapnoisy}), based on the Lyapunov function
\begin{align}\nonumber
G(\mathbf W)\triangleq f(\mathbf W)+\frac{\gamma\Lambda^*}{2\eta}{\mathbf W}^\top(\mathbf I_{Nd}-\W){\mathbf W},
\end{align}
\vspace{-4mm}\\
where the second term incentivizes consensus (in fact, it equals zero when $\mathbf w_i=\mathbf w_j,\forall i,j$).
We can then rewrite
\begin{align}
\label{Lyapnoisy}
\mathbf W_{k+1}=\Pi^N[\mathbf W_{k}
-\eta\nabla G(\mathbf W_{k})
+\gamma\e_k
].
\end{align}
\vspace{-4mm}\\
\indent We study the convergence of \proposed\ under the following standard assumptions.
\begin{assumption}
\label{fiassumption}
 $f_i(\mathbf w)$ are  $\mu$-strongly convex and $L$-smooth.
\end{assumption}
\begin{assumption}
\label{distance}
$\zeta{\triangleq}R{-}\Vert\mathbf w^*\Vert>0$
 ($\mathbf w^*$ is in the interior of $\mathcal W$).
\end{assumption}
Since
 $\boldsymbol{\Omega}$ is symmetric and doubly-stochastic, 
its eigenvalues ($\rho_i$ for the $i$th one) are real-valued and
$1{=}\rho_1{\geq}\rho_2{\geq}{\dots}{\geq}\rho_N\geq-1$.
 We make the following standard assumption on $\rho_2$.
\begin{assumption} 
 $
\rho_2<1.
$
\end{assumption}

We are now ready to present the convergence properties of \proposed. The main idea is to decompose the error between the local optimization variables and the global optimum into: 1) the error between $\mathbf W_k$ and the minimizer of the Lyapunov function $G$, defined as
\vspace{-1mm}
\begin{align}
\label{penalized}
\mathbf W^{(G)}=\arg\min_{\mathbf W\in\mathcal W^N} G(\mathbf W);
\end{align}
and 2) the error between the latter and the global optimum $\mathbf w^*$.
Accordingly, 
we bound via the triangle inequality
\begin{align}
\label{int30}
\nonumber
&
\Big(\mathbb E\Big[\frac{1}{N}\sum_{i=1}^N\Vert{\mathbf w}_{i,k}-\mathbf w^*\Vert^2\Big]\Big)^{1/2}=
\frac{1}{\sqrt{N}}\Vert{\mathbf W}_{k}-\mathbf 1\otimes\mathbf w^*\Vert_{\mathbb E}
\nonumber
\\&
\leq\frac{1}{\sqrt{N}}\Vert{\mathbf W}_{k}-\mathbf W^{(G)}\Vert_{\mathbb E}
+\frac{1}{\sqrt{N}}\Vert\mathbf W^{(G)}-\mathbf 1\otimes\mathbf w^*\Vert.
\end{align}
These terms are individually bounded in Theorem \ref{Tmain}. A sketch of its proof is provided  in the Appendix.
 \begin{theorem}\label{Tmain} 
 Assume: ({\bf C1}) $\eta(\mu{+}L){+}\gamma\Lambda^*(1{-}\rho_N){\leq} 2$; ({\bf C2}) $\frac{\eta}{\gamma}{\leq}\frac{\zeta\cdot Z}{\sqrt{N}\nabla_{\max}}$, with $Z{\triangleq}\frac{(1{-}\rho_2)\Lambda^*}{2\sqrt{1{+}L/\mu}}$, $\nabla_{\max}{\triangleq}\underset{i}{\max}\Vert\nabla f_{i}(\mathbf w^*)\Vert$.
Then, 
\begin{align}
\label{L1}
&\!\!\!\!\frac{1}{\sqrt{N}}\Vert{\mathbf W}_{k}{-}\mathbf W^{(G)}\Vert_{\mathbb E}
{\leq} 2R\Big[\frac{\sqrt{2}d}{\sqrt{\mu}}\Big(\Lambda^*{+}\frac{\sigma^2}{E}\Big)\frac{\gamma}{\sqrt{\eta}}{+}e^{-\mu\eta k}\Big],
\\
&\frac{1}{\sqrt{N}}\Vert\mathbf W^{(G)}-\mathbf 1\otimes\mathbf w^*\Vert\leq \frac{\nabla_{\max}}{Z}\frac{\eta}{\gamma}.
\label{L4}
\end{align}
 \end{theorem}

To minimize these errors, $\gamma/\sqrt{\eta}$ and $\eta/\gamma$ need both be small,
while $\eta$ should be large to make $e^{-\mu\eta k}$ small, yielding a tradeoff between the tuning of 
$\eta$ and $\gamma$.
To further investigate the convergence properties, let us consider a target timeframe $K$ at which the algorithm stops.
It appears then reasonable to choose $\eta= a\cdot K^{-x}$ and $\gamma=b\cdot K^{-y}$
for suitable $a,b,x,y>0$. Under this choice,  \eqref{L1}-\eqref{L4} specialize as
\begin{align*}
\nonumber
&\frac{1}{\sqrt{N}}\Vert{\mathbf W}_{K}{-}\mathbf W^{(G)}\Vert_{\mathbb E}
\leq 
2R\Big[\frac{\Lambda^*{+}\frac{\sigma^2}{E}}{\sqrt{a\mu}}\sqrt{2}dbK^{\frac{x}{2}-y}{+}e^{-\mu a\cdot K^{1-x}}\Big]
,\\&
\frac{1}{\sqrt{N}}\Vert\mathbf W^{(G)}-\mathbf 1\otimes\mathbf w^*\Vert\leq \frac{\nabla_{\max}}{Z} \frac{a}{b}K^{y-x}.
\end{align*}
The exponential term requires $x{<}1$ to converge when $K{\to}\infty$, while
 the other two terms are of order $\mathcal O(K^{\max\{y-x,x/2-y\}})$. Hence,
 $\max\{y-x,x/2-y\}$ should be minimized subject to $x<1$, yielding the following corollary.
 \begin{corollary}
 \label{corol}
 Let $0{<}\epsilon{<}1$.
 With $\eta{\propto} K^{-(1-\epsilon)}$ and $\gamma{\propto} K^{-3/4(1-\epsilon)}$, we have $e^{-\mu a\cdot K^{\epsilon}}{=}\mathcal O(K^{-1/4(1-\epsilon)})$ and
 \vspace{-2mm}
 $$
 \Big(\mathbb E\Big[\frac{1}{N}\sum_{i=1}^N\Vert{\mathbf w}_{i,k}-\mathbf w^*\Vert^2\Big]\Big)^{1/2}
 =\mathcal O(K^{-1/4(1-\epsilon)}).
 $$
 \end{corollary}
  \vspace{-2mm}
When $\epsilon\to 0$, we can see that the error scales as $\mathcal O(K^{-1/4})$, which is also validated numerically in the next section.
\vspace{-2mm}
\section{Numerical Results}
\label{numres}
We solve the '0 versus 1' task based on the MNIST dataset \cite{MNIST}:
 the goal is to
 distinguish images of digits '0' and '1'.

\emph{Network deployment}: We consider  $N{=}200$ nodes, spread uniformly at random over a region of $3$km radius. The nodes communicate over a bandwidth of $W_{tot}{=}$1MHz, carrier frequency $f_c=3$GHz,
 with a fixed transmission power of $P_{tx}=5$dBm. 
The noise power spectral density at the receivers is $N_0=-169$dBmW/Hz.
The average pathloss $\Lambda_{i,j}$ between node $i$ and $j$ follows Friis'  free space equation.
 
\emph{Data deployment}:
Each node has a local dataset with a single 28x28 pixels image: 100 nodes have digit '0', the remainder have digit '1'.
Node $i$'s image is converted into a 50-dimensional real feature vector
 $\mathbf d_i$,
representing the components (out of a total of $28\times 28=784$) with largest mean energy across the dataset,
  and then
normalized to $\Vert\mathbf d_i\Vert=1$.
We define the label $\ell_i{=}1$ if node $i$'s image is of digit 0, otherwise $\ell_i{=}-1$.
 
 \emph{Optimization problem formulation}:
 We solve the task via regularized logistic regression, with loss function 
\begin{align}
\label{logreg}
f_i(\mathbf w)=\frac{0.01}{2}\Vert\mathbf w\Vert^2+\ln\left(1+\exp\{
-\ell_i\cdot\mathbf d_i^\top\mathbf w\}\right),
\end{align}
 where $\mathbf w\in\mathbb R^d$ is a $d=50$-dimensional parameter vector.
It can be shown that  $f_i(\mathbf w)$, hence the global function $F(\mathbf w){=}\frac{1}{N}\sum_{i=1}^Nf_i(\mathbf w)$,
are all strongly-convex with parameter $\mu{=}0.01$, and smooth with parameter $L=\mu+1/4$.

\emph{Wireless distributed optimization algorithms}:
We implement the following algorithms, all initialized as $\mathbf w_{i,0}=\mathbf 0,\forall i$.

$\bullet$ \underline{\proposed} (proposed): 
To enforce half-duplex constraints, 
100 nodes, selected randomly, transmit in slot one, the others transmit in slot two.
The frame duration is
$T
=102\mu$s.

We also compare the proposed \proposed\ with implementations of DGD over orthogonal digital (OD-DGD) and analog (OA-DGD) channels. Both follow the updates
\begin{align}
\label{ODGD}
\mathbf w_{i,k+1}
&=\Pi[\mathbf c_{i,k}-\eta\nabla f_i(\mathbf w_{i,k})],
\end{align}
where $\mathbf c_{i,k}$ is a reconstruction of $\sum_{j=1}^N\omega_{i,j}\mathbf w_{j,k}$, but  differ in how signals are encoded and transmitted, and $\mathbf c_{i,k}$ is computed:

\begin{figure*}
     \centering
     \hfill
     \begin{subfigure}[b]{0.32\linewidth} 
        \includegraphics[width = \linewidth,trim=15 0 35 20, clip=true]{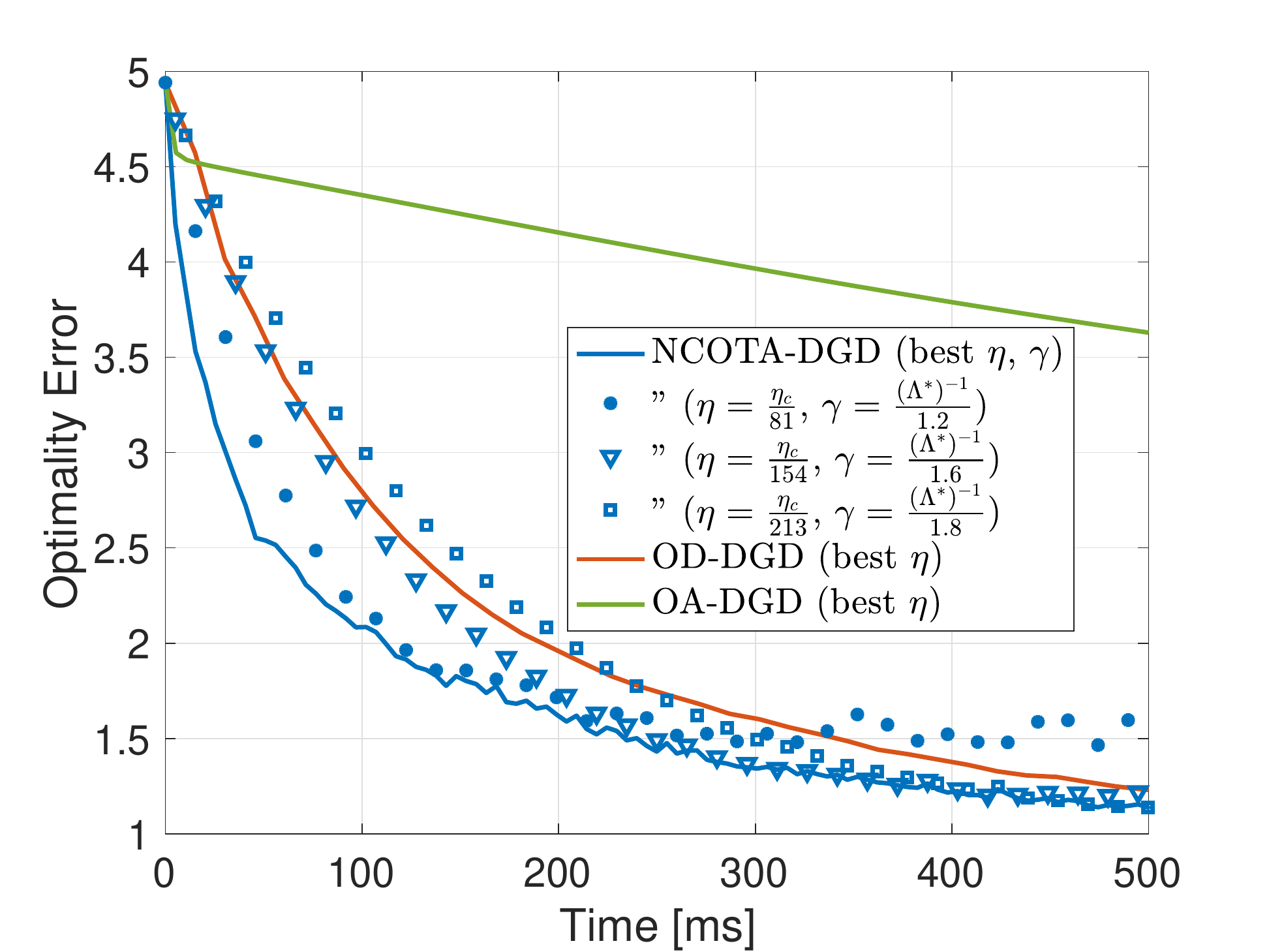}
	    \caption{} \label{fig:vstime}
     \end{subfigure}
     \hfill
     \begin{subfigure}[b]{0.32\linewidth}
\includegraphics[width = \linewidth,trim=15 0 35 20, clip=true]{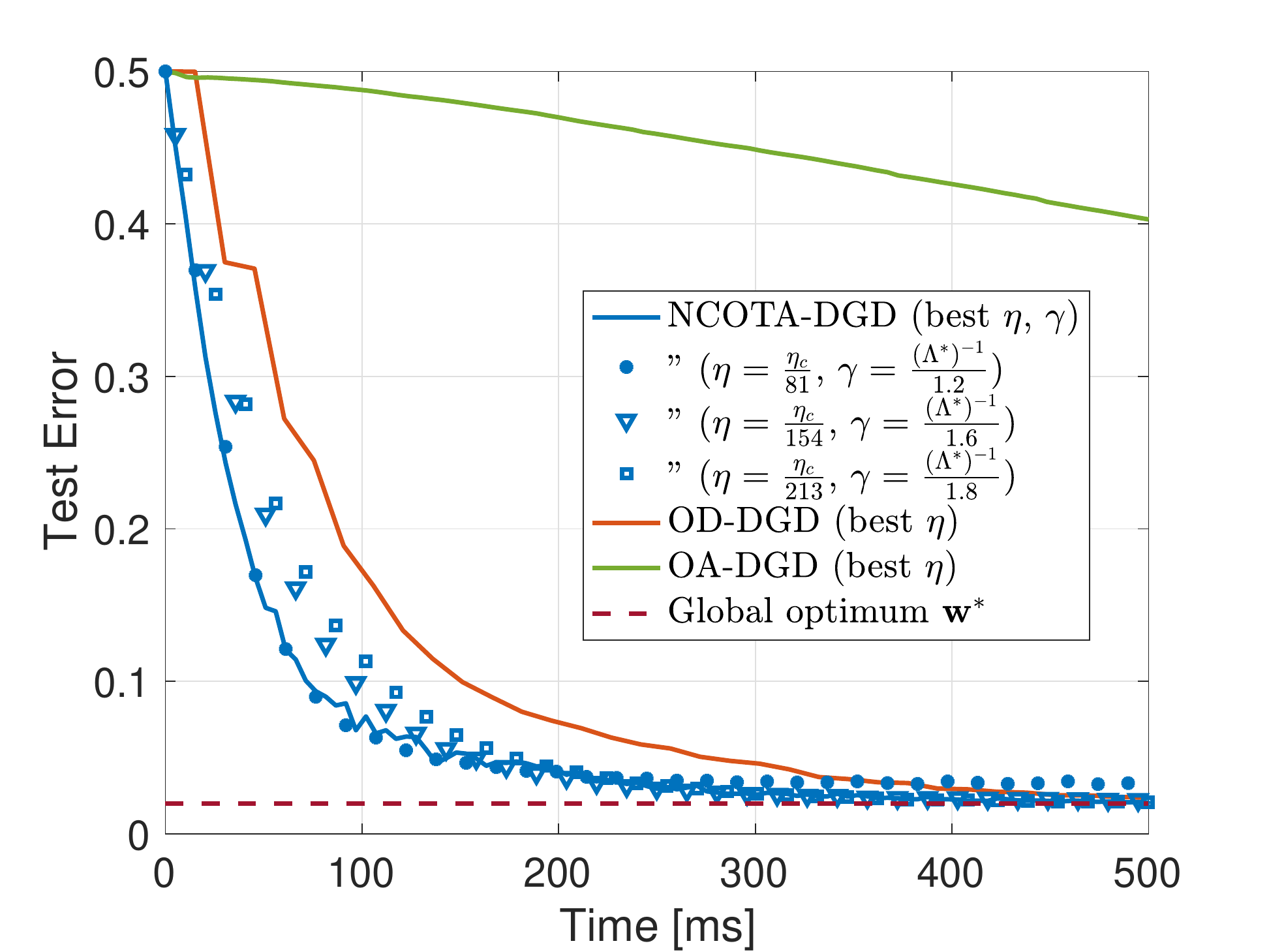}
	    \caption{} \label{fig:Testvstime}
     \end{subfigure}
          \hfill
     \begin{subfigure}[b]{0.34\linewidth}
\includegraphics[width = \linewidth,trim=0 0 0 0, clip=true]{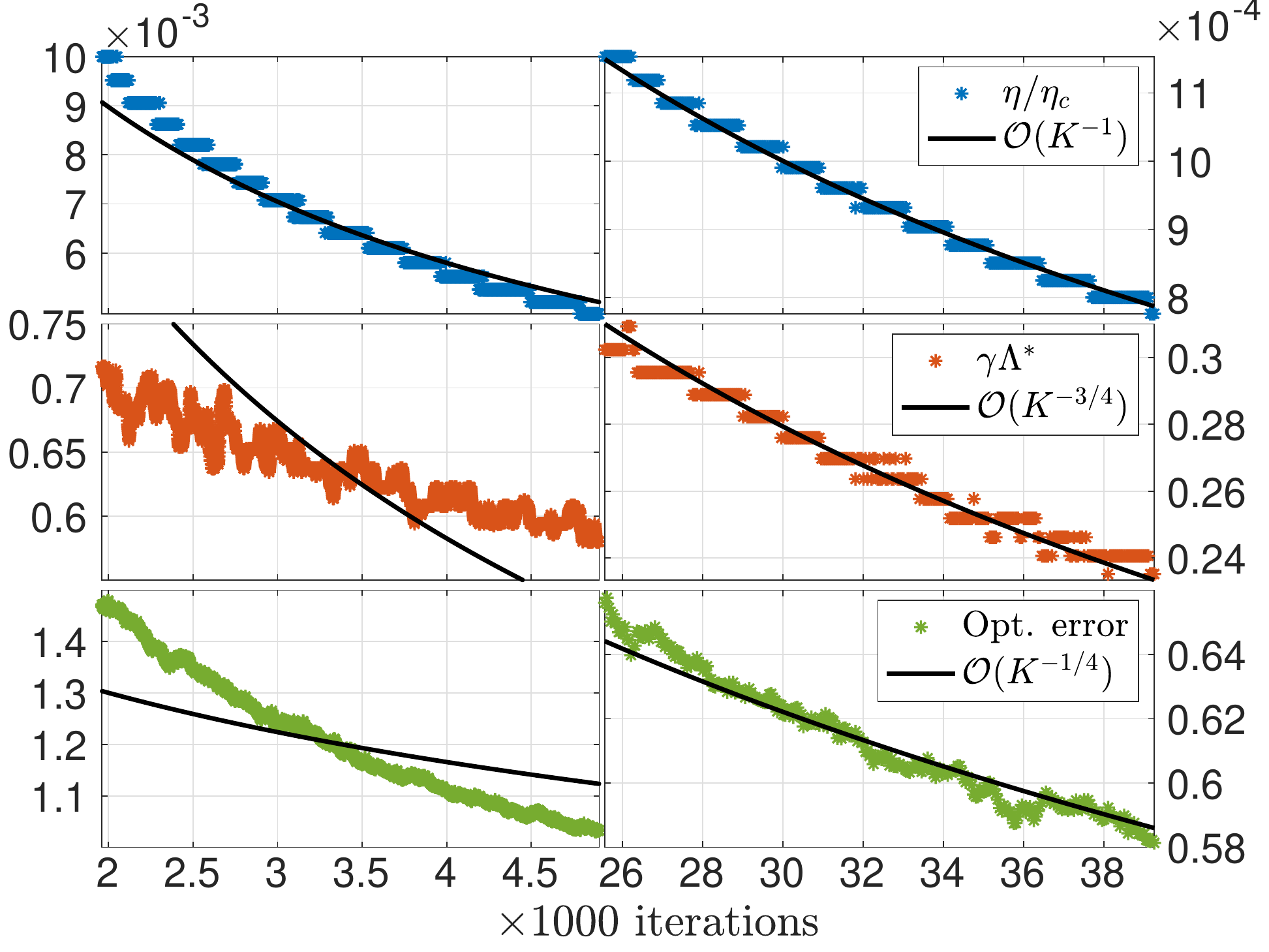}
	    \caption{} \label{fig:fitcurve}
     \end{subfigure}
     \hfill
     \vspace{-4mm}
\caption{Optimality error (a), test error (b), vs running time. Optimal stepsizes and optimality error for \proposed, vs iterations (c).\vspace{-7mm}} 
\end{figure*}

$\bullet$ \underline{Orthogonal Digital DGD (OD-DGD)}: 
each node scales $\mathbf w_{i,k}$ by the largest magnitude of its components, $\Vert\mathbf w_{i,k}\Vert_{\infty}$;
each component of  $\mathbf w_{i,k}/\Vert\mathbf w_{i,k}\Vert_{\infty}$ (each $\in[-1,1]$) is then quantized using 
9 quantization levels uniformly spaced in the interval $[-1,1]$.
We use dithered quantization: letting $\hat{\mathbf w}_{i,k}$ be the quantized signal, it is such that $\mathbb E[\hat{\mathbf w}_{i,k}|{\mathbf w}_{i,k}]={\mathbf w}_{i,k}$.
With $\Vert\mathbf w_{i,k}\Vert_{\infty}$ encoded using machine precision (64 bits), the data payload is thus
$L=64+d\log_2(9)\approx 223$bits to encode the $d$-dimensional signal $\mathbf w_{i,k}$.
Such payload is then transmitted by each node over orthogonal channels (via TDMA), using capacity achieving codes with rate $R$. With the fading channel $h_{i,j}^k\sim\mathcal{CN}(0,\Lambda_{i,j})$ between transmitting node $j$ and receiving node $i$, and assuming CSI at the receiver,
the probability of successful decoding is
$P_{i,j}^{succ}{\triangleq}\exp\{-\frac{\sigma^2}{E\Lambda_{i,j}}(2^R-1)\}$. 
$R$ is chosen to guarantee a minimum $90\%$ success probability for nodes within a 500m radius from the transmitting node, yielding
$R{\approx}2$bits/s/Hz.
The resulting frame duration is $T
\approx  22.67$ms.
At the end of the $N$ transmissions, node $i$ computes
\vspace{-2mm}
$$
\mathbf c_{i,k}=\mathbf w_{i,k}
+\frac{1}{\max_n\sum_{j\neq n}P_{n,j}^{succ}}
\sum_{j\neq i}\iota_{i,j}^k(\hat{\mathbf w}_{j,k}-\mathbf w_{i,k}),
\vspace{-2mm}
$$
where $\iota_{i,j}^k{=}\mathbbm{1}[R{<}\log_2(1+|h_{i,j}^k|^2E/\sigma^2)]$ indicates a successful reception of $\hat{\mathbf w}_{j,k}$ at node $i$.
With this choice of $\mathbf c_{i,k}$, the updates \eqref{ODGD} represent a noisy version of 
\eqref{ddd} with weights $\omega_{i,j}{=}\frac{P_{i,j}^{succ}}{\max_n\sum_{j\neq n}P_{n,j}^{succ}},j{\neq}i$, $\omega_{i,i}{=}1{-}\sum_{j\neq i}\omega_{i,j}$, and the additional projection step (as seen by computing $\mathbb E[\mathbf c_{i,k}|\mathcal F_k]$).

$\bullet$ \underline{Orthogonal Analog DGD (OA-DGD)}: 
$\mathbf w_{i,k}$ is first normalized to unit norm;
the first (respectively, second) half of the normalized vector, $\frac{[\mathbf w_{i,k}]_{1:d/2}}{\Vert\mathbf w_{i,k}\Vert}$ ($\frac{[\mathbf w_{i,k}]_{d/2+1:d}}{\Vert\mathbf w_{i,k}\Vert}$) is mapped to the real (imaginary) part of the baseband transmitted signal as
$$
\mathbf x_{i,k}=\sqrt{E\frac{d/2+2}{3}}\Big[\frac{[\mathbf w_{i,k}]_{1:\frac{d}{2}}{+}\mathrm{j}[\mathbf w_{i,k}]_{\frac{d}{2}+1:d}}{\Vert\mathbf w_{i,k}\Vert};
\frac{\Vert\mathbf w_{i,k}\Vert}{R};1\Big];
$$
note that $\mathbf x_{i,k}$
 includes the norm $\Vert\mathbf w_{i,k}\Vert$ (penultimate sample) and a pilot signal (last sample) to estimate the channel at the receiver.
 This constitutes a $(d/2{+}2)$-dimensional  complex-valued signal, whose energy per sample satisfies
$\frac{1}{d/2+2}\Vert\mathbf x_{i,k}\Vert^2\leq E,$
consistent with the power constraint.
The signal is then transmitted by each node over orthogonal channels (via TDMA), yielding the frame duration
$T=
5.4$ms.
With the received signal
$
\mathbf y_{i,j}^k=h_{i,j}^k\mathbf x_{j,k}+\mathbf n_{i,j}^k,
$
node $i$ first estimates $h_{i,j}^k$ via maximum likelihood from the last sample, followed by the estimation of
$\Vert\mathbf w_{j,k}\Vert$ from the penultimate sample; it then estimates $\mathbf w_{j,k}/\Vert\mathbf w_{j,k}\Vert$ from the first $d/2$ samples.
After receiving the signals from all nodes, and using the reconstructions $\hat{\mathbf w}_{j,k}$ of ${\mathbf w}_{j,k}$, node $i$ then computes
\vspace{-1mm}
$$
\mathbf c_{i,k}=\mathbf w_{i,k}
+\sum_{j\neq i}\frac{\Lambda_{i,j}}{\max_n\sum_{j\neq n}\Lambda_{n,j}}(\hat{\mathbf w}_{j,k}-\mathbf w_{i,k}),
\vspace{-2mm}
$$
so that signals are mixed proportionally to the average pathloss.
With this choice of $\mathbf c_{i,k}$, the updates \eqref{ODGD} represent a noisy version of 
\eqref{ddd} with weights $\omega_{i,j}=\frac{\Lambda_{i,j}}{\max_n\sum_{j\neq n}\Lambda_{n,j}},\ j\neq i$ and $\omega_{i,i}=1-\sum_{j\neq i}\omega_{i,j}$, and the additional projection step.

Note that OD-DGD requires CSI at the receiver and knowledge of $P^{succ}$; OA-DGD requires knowledge of the average pathloss $\Lambda_{i,j}$. 
In this simulation, we idealistically assume that such information is available at no cost.
In contrast, the proposed \proposed\ does not require such knowledge.

\emph{Evaluations and Discussion}:
We evaluate:
(1) the \emph{optimality error} $\sqrt{\hat{\mathbb E}[\frac{1}{N}\sum_{i=1}^N\Vert\mathbf w_{i,k}{-}{\mathbf w}^*\Vert^2]}$,
measuring the deviation of the local models from the solution of \eqref{global} (bounded in expectation in Theorem \ref{Tmain});
(2) the \emph{average test error} $\mathrm{TEST}_k=\hat{\mathbb E}[\frac{1}{N}\sum_{i=1}^N \mathrm{TEST}_{i,k}]$, where
$\mathrm{TEST}_{i,k}$ is the test error for node $i$ at frame $k$. This is computed on a test set of  100 '0's and 100 '1's; the label associated to feature vector $\mathbf d$ is predicted as '0' if $\mathbf w_{i,k}^\top\mathbf d{>}0$, and '1' otherwise.
Here, $\hat{\mathbb E}[\cdot]$ denotes a sample average of 10 trajectories generated by the algorithms, over independent realizations of fading and noise.

In Fig.~\ref{fig:vstime}, we plot
the \emph{best} optimality error
 vis-\`a-vis running time:
 all algorithms are evaluated using a set of fixed stepsizes; for each time instance in the $x$-axis, we plot 
 only the  optimality error evaluated on the best performing stepsize choice at that time. For \proposed, we also show the curves for 3 representative stepsize choices.
 \proposed\ achieves the best performance, followed by OD-DGD and OA-DGD,
thanks to its fast updates: during 500ms, \proposed\ performs $4900$ iterations, versus only 22 of OD-DGD
 and 93 of OA-DGD, which are both limited by the use of orthogonal channels.
 Yet, OD-DGD bridges the gap to \proposed\ over time: this is due to its better noise-suppression capabilities--especially beneficial when approaching convergence to $\mathbf w^*$.
This behavior suggests that a mixed analog-digital strategy may further improve performance--a study left for future work.
 OA-DGD performs the worst:
 it does not enjoy the noise suppression capabilities of OD-DGD, and its updates are $53\times$ slower than \proposed.
 In Fig.~\ref{fig:Testvstime}, we plot
the test error under the same best stepsize choice, vis-\`a-vis running time. It follows a similar trend as Fig.~\ref{fig:vstime}.
Both \proposed\ and OD-DGD converge to the test error under the optimal $\mathbf w^*$.

 In Fig.~\ref{fig:fitcurve}, 
  we plot 
  the best stepsizes $\eta$ (first row), $\gamma$ (second row) and
   best optimality error (third row) of \proposed, for two different phases: 
   \emph{initial phase}, corresponding to the 
   first $\sim$5k iterations ($500$ms, left side); \emph{asymptotic phase}, after 26k iterations (right side). For each phase, we also fit the data points to the theoretical stepsize tuning and convergence behavior found in Corollary \ref{corol}: 
   solid lines correspond to $\eta{\propto}(K{+}\delta)^{{-}1}$, $\gamma{\propto}(K{+}\delta)^{{-}3/4}$ and "Opt. error"${\propto}(K{+}\delta)^{{-}1/4}$,
where the scaling factors ($\propto$) and $\delta$ are fit to the data points.
   We note that, in the initial phase, the optimal $\eta$ and $\gamma$ do not match the theoretical behavior. In fact, in this regime, the optimality error is dominated by \eqref{L4}, 
  and decreases quicker than $\mathcal O(K^{-1/4})$ (bottom left).
  Conversely,  in the asymptotic phase, the optimal $\eta$ and $\gamma$ more closely match the theoretical scaling,
     and the optimality error decays as ${\propto}(K{+}\delta)^{{-}1/4}$, as predicted
    (with $\delta=4227$). This is in line with Corollary \ref{corol} when ${\epsilon}{\to} 0$, and corresponds to a regime when both error terms \eqref{L1}-\eqref{L4} are equally dominant.
\section{Conclusions}
\label{conclu}
We presented \proposed, an implementation of DGD that 
combines local gradient descent with a novel Non-Coherent Over-The-Air consensus scheme to solve 
distributed machine-learning problems over wirelessly-connected systems. 
\proposed\  enables simultaneous transmissions by mapping local optimization signals to a mixture of preamble sequences, and consensus by correlating the received signals with the preamble sequences via non-coherent combining, without explicit knowledge of the mixing weights, nor channel state information.
We proved its convergence properties, both theoretically and numerically, and showed superior performance than
 implementations of DGD over  digital and analog orthogonal channels.

\vspace{0mm}
\appendix
\renewcommand\thesubsection{\thesection.\Roman{subsection}}

\def\thesubsectiondis{Appendix \thesection.\Roman{subsection}:} 

\begin{proof}[Proof sketch of Theorem \ref{Tmain}]
To prove \eqref{L1}, we use
 the fixed-point optimality condition $\mathbf W^{(G)}{=}\Pi^N[\mathbf W^{(G)}
{-}\eta\nabla G(\mathbf W^{(G)})]$, the non-expansive property of projections \cite{bertsekas2003convex},
and take the expectation conditional on $\mathcal F_k$, yielding $\Vert{\mathbf W}_{k+1}-\mathbf W^{(G)}\Vert_{\mathbb E|\mathcal F_k}^2$
$$
{\leq}\Vert\mathbf W_{k}-\mathbf W^{(G)}
{-}\eta(\nabla G(\mathbf W_{k}){-}\nabla G(\mathbf W^{(G)}))
\Vert^2
{+}\gamma^2\Sigma,
$$
where $\Vert\e_k\Vert_{\mathbb E|\mathcal F_k}^2\leq \Sigma$ (Lemma \ref{L0} at the end of the Appendix).
Assumption \ref{fiassumption} implies that $G$ is $\mu$-strongly convex and $L_{G}\triangleq L+\Lambda^*(1-\rho_N)\gamma/\eta$-smooth.
\balance
Then, using \cite[Theorem 2.1.12]{Nesterov2004}, the first term above is further bounded as
$
{\leq}
(1{-}\mu\eta)^2\Vert\mathbf W_k{-}\mathbf W^{(G)}\Vert^2
$
as long as $\eta\leq 2/(\mu+L_{G})$ (equivalent to {\bf C1}),
yielding, after the unconditional expectation, 
$$
\Vert\mathbf W_{k+1}-\mathbf W^{(G)}\Vert_{\mathbb E}^2
\leq
(1-\mu\eta)^2\Vert\mathbf W_{k}-\mathbf W^{(G)}\Vert_{\mathbb E}^2
+\gamma^2\Sigma.
$$
The result \eqref{L1} follows
 after solving the induction, noting that $\Vert\mathbf W_{0}-\mathbf W^{(G)}\Vert\leq \sqrt{N}2R$,
 using the expression of $\Sigma$ in Lemma \ref{L0},
$\mu\eta\leq 1$ (implied by {\bf C1}), $(1-\mu\eta)^{2k}\leq e^{-2\mu\eta k}$,
  and $\sqrt{a+b}\leq\sqrt{a}+\sqrt{b}$ for $a,b\geq 0$.  

Next, we prove \eqref{L4}.
Consider the \emph{unconstrained} minimizer
$
\hat {\mathbf W}\triangleq \arg\min_{\mathbf W\in\mathbb R^{Nd}}G(\mathbf W).
$
Hence, $\mathbf 0{=}\nabla G(\hat {\mathbf W})$ ${=}\nabla f(\hat {\mathbf W}){+}\frac{\gamma}{\eta}\Lambda^*(\mathbf I_{Nd}{-}\W)\hat {\mathbf W}$.
From the multivariate mean value theorem, there exists $\mathbf A$ with $\mu\mathbf I_{Nd}{\preceq}\mathbf A{\preceq} L\mathbf I_{Nd}$ such that 
$
\nabla f(\hat {\mathbf W}){=}\nabla f(\mathbf 1{\otimes}\mathbf w^*){+}\mathbf A(\hat {\mathbf W}{-}\mathbf 1{\otimes}\mathbf w^*).
$
Combining it with $\nabla G(\hat {\mathbf W}){=}\mathbf 0$ yields
$
\Vert\hat {\mathbf W}{-}\mathbf 1\otimes\mathbf w^*\Vert
{=}
\Vert\mathbf B^{-1}\nabla f(\mathbf 1\otimes\mathbf w^*)\Vert,
$
where
$\mathbf B{\triangleq}\mathbf A{+}\frac{\gamma}{\eta}\Lambda^*(\mathbf I_{Nd}{-}\W){\succeq}\mu\mathbf I_{Nd}.$
Note that the optimality condition on $\mathbf w^*$, $\sum_{i=1}^N\nabla f_i(\mathbf w^*){=}0$, implies
 $\nabla f(\mathbf 1{\otimes}\mathbf w^*)\bot(\mathbf 1{\otimes}\mathbf I_d)$; hence we further bound
$$
\Vert\hat {\mathbf W}-\mathbf 1\otimes\mathbf w^*\Vert
\leq
\Vert\nabla f(\mathbf 1\otimes\mathbf w^*)\Vert
\max_{\mathbf v\bot(\mathbf 1\otimes\mathbf I_d):\Vert\mathbf v\Vert=1}\Vert\mathbf B^{-1}\mathbf v\Vert.
$$
Furthermore, 
 $\Vert\nabla f(\mathbf 1\otimes\mathbf w^*)\Vert\leq \sqrt{N}\nabla_{\max}$,
and it can be proved (not shown due to space constraints, see \cite{techreport}) that 
$$
\max_{\mathbf v\bot(\mathbf 1\otimes\mathbf I_d):\Vert\mathbf v\Vert=1}
\Vert\mathbf B^{-1}\mathbf v\Vert
\leq
\frac{
2\sqrt{1+L/\mu}}{\Lambda^*(1-\rho_2)}\frac{\eta}{\gamma},
$$
yielding
$
\Vert\hat {\mathbf W}{-}\mathbf 1{\otimes}\mathbf w^*\Vert
{\leq}
\frac{\sqrt{N}\nabla_{\max}}{Z}\frac{\eta}{\gamma}.
$
Next, we show that $\hat {\mathbf W}{\in}\mathcal W^N$, hence it coincides with
${\mathbf W}^{(G)}$ solution of the constrained problem. 
Since $\mathbf w^*$ is at distance $\zeta$ from the boundary of $\mathcal W$ (Assumption \ref{distance}), it  suffices to show that
$\mathbf w^*$ is closer to $\hat {\mathbf w}_{i}$ than to the boundary of $\mathcal W$, i.e.
 $\Vert\hat {\mathbf w}_{i}{-}\mathbf w^*\Vert{\leq}\zeta,\ \forall i$. This is a direct consequence of {\bf C2}:
$
\Vert\hat {\mathbf w}_{i}{-}\mathbf w^*\Vert{\leq}\Vert\hat {\mathbf W}{-}\mathbf 1{\otimes}\mathbf w^*\Vert
$
$
{\leq}\frac{\sqrt{N}\nabla_{\max}}{Z}\frac{\eta}{\gamma},
$
hence
 $\hat {\mathbf W}{=}{\mathbf W}^{(G)}$ 
and \eqref{L4} follows.
\end{proof}

\vspace{-2mm}
\begin{lemma}
\label{L0}
$
\Vert\e_k\Vert_{\mathbb E|\mathcal F_k}^2\leq
8N[Rd(\Lambda^*+\sigma^2/E)]^2
\triangleq
 \Sigma.
$
\end{lemma}
\begin{proof}[Proof sketch]
Using \eqref{dik},
we rewrite $\e_{i,k}$ as
\vspace{-1mm}
$$
\e_{i,k}=
\sum_{m=1}^M(|r_{i,m}^{k}|^2-\mathbb E[|r_{i,m}^{k}|^2|\mathcal F_k])(\mathbf z_m-{\mathbf w_{i,k}}).
$$
Using the triangle inequality,
 we bound
 $$
\Vert\e_{i,k}\Vert_{\mathbb E|\mathcal F_k}\leq
\sum_{m=1}^M\mathrm{sd}(|r_{i,m}^{k}|^2|\mathcal F_k)
\Vert\mathbf z_m-{\mathbf w_{i,k}}\Vert,
$$
where $\mathrm{sd}(|r_{i,m}^{k}|^2|\mathcal F_k)$ 
is the standard deviation of $|r_{i,m}^{k}|^2$, conditional on $\mathcal F_k$; since
$|r_{i,m}^{k}|^2|\mathcal F_k$ is exponentially distributed, it 
equals $\mathbb E[|r_{i,m}^{k}|^2|\mathcal F_k]$.
Moreover, 
$\Vert\mathbf z_m{-}{\mathbf w_{i,k}}\Vert{\leq}\max_{m,m'}\Vert\mathbf z_m{-}\mathbf z_{m'}\Vert{=}\sqrt{8}Rd$.
Using \eqref{Erm}, it then follows
$\Vert\e_{i,k}\Vert_{\mathbb E|\mathcal F_k}{\leq}\sqrt{8}Rd[\sum_{j\in\mathcal N_i}\Lambda_{i,j}{+}\frac{\sigma^2}{E}].
$
The result directly follows after using $\sum_{j\in\mathcal N_i}\Lambda_{i,j}\leq \Lambda^*$, squaring both sides and adding over $i$ ($\Vert\e_{k}\Vert_{\mathbb E|\mathcal F_k}^2
=
\sum_{i=1}^N\Vert\e_{i,k}\Vert_{\mathbb E|\mathcal F_k}^2$).
\end{proof}

 \vspace{-2mm}
   \bibliographystyle{IEEEtran}
\bibliography{IEEEabrv,biblio}

\end{document}